\documentclass[12pt]{article}

\usepackage[T1]{fontenc}
\usepackage[sc]{mathpazo}
\usepackage{amsmath}
\usepackage{amssymb}
\usepackage{enumerate}
\usepackage{amsthm}
\usepackage{amsfonts,mathrsfs}

\usepackage{hyperref}
\hypersetup{pdfpagemode=UseNone}

\newenvironment{mylist}[1]{\begin{list}{}{
    \setlength{\leftmargin}{#1}
    \setlength{\rightmargin}{0mm}
    \setlength{\labelsep}{2mm}
    \setlength{\labelwidth}{8mm}
    \setlength{\itemsep}{0mm}}}
    {\end{list}}

\newcommand{\defeq}{\stackrel{\smash{\textnormal{\tiny def}}}{=}}

\newcommand{\Pa}[1]{\left(#1\right)}

\newcommand{\Br}[1]{\left[#1\right]}

\newtheorem{thrm}{Theorem}[section]

\theoremstyle{definition}
\newtheorem{definition}[thrm]{Definition}
\newtheorem{remark}[thrm]{Remark}

\numberwithin{equation}{section}

\newcounter{questionnumber}

\begin{document}

\vspace*{-8mm}


\begin{center}

{\Large\textbf{Fermat-linked relations for the Boubaker polynomial  sequences via Riordan matrices analysis}}\\[6mm]
{\large Karem Boubaker\footnote{E-mail: mmbb11112000@yahoo.fr}}\\
{\it\footnotesize ESSTT/ 63 Rue Sidi Jabeur 5100, Mahdia, Tunisia}\\[4mm]
{\large Lin Zhang\footnote{E-mail: godyalin@163.com}}\\
{\it\footnotesize Institute of Mathematics, Hangzhou Dianzi University,
Hangzhou, 310037, P.~R.~China}

\end{center}


\begin{abstract}

The Boubaker polynomials are investigated in this paper. Using Riordan matrices analysis, a sequence of relations outlining the relations with Chebyshev and Fermat polynomials have been obtained. The obtained expressions are a meaningful supply to recent applied physics studies using the Boubaker polynomials expansion scheme (BPES).

\end{abstract}

\section{Introduction}

Polynomial expansion methods are extensively used in many mathematical and engineer fields to yield meaningful results for both numerical and analytical analysis \cite{Alvareza,Bender,Choi,Guertz,Koelink,Okada,Philippou,Sloan}. Among the most frequently used polynomials, the Boubaker polynomials are one of the interesting tools which were associated to several applied physics problems as well as the related polynomials such as the Boubaker-Turki polynomials \cite{Boubaker1,Bagula,Sloane,Ghanouchi1,Ghanouchi2,Slama2,Slama3,Ghrib,Boubaker2}, the $4-q$ Boubaker polynomials \cite{Zhao1} and the Boubaker-Zhao polynomials \cite{Zhao2}. For example, for some resolution purposes, a function $f(r)$ is expressed as an infinite nonlinear expansion of Boubaker-Zhao polynomials
\begin{eqnarray}
f(r) = \lim_{N\to+\infty}\Br{\frac{1}{2N}\sum^N_{n=1}\zeta_n \hat{B}_{4n}\Pa{r\frac{\alpha_n}{R}}},
\end{eqnarray}
where $\alpha_n$ are the minimal positive roots of the Boubaker $4n$-order polynomials $\hat{B}_{4n}$, $R$ is a maximum radial range and $\zeta_n$ are coefficients to be determined using the expression of $f(r)$. Since the Boubaker $4n$-order polynomials have the particular properties: for any $n$,
\begin{eqnarray}
\left\{\begin{array}{rcl}
         \hat{B}_{4n}(r)\left|\right._{r=0} &=& -2   \\
         \frac{\partial \hat{B}_{4n}(r)}{\partial r} &=& 0  \\
         \frac{\partial^2 \hat{B}_{4n}(r)}{\partial r^2} &=& 4n(n-1)
       \end{array}
\right.
\end{eqnarray}
The related the system (\ref{eq:1.3}) is induced:
\begin{eqnarray}\label{eq:1.3}
\left\{\begin{array}{rcl}
         f(0)&=& \lim_{N\to+\infty}\Br{\frac1{2N}\sum^N_{n=1}\zeta_n\hat{B}_{4n}\Pa{r\frac{\alpha_n}{R}}}\left|\right._{r=0} = -\frac1N\sum^N_{n=1}\zeta_n\\
         f(R)&=& \lim_{N\to+\infty}\Br{\frac1{2N}\sum^N_{n=1}\zeta_n\hat{B}_{4n}\Pa{r\frac{\alpha_n}{R}}}\left|\right._{r=R} = 0\\
         \frac{\partial f(r)}{\partial r}\left|\right._{r=0}&=& \lim_{N\to+\infty}\Br{\frac1{2N}\sum^N_{n=1}\zeta_n\frac{\partial\Pa{ \hat{B}_{4n}\Pa{r\frac{\alpha_n}{R}}}}{\partial r}}\left|\right._{r=0} = 0
       \end{array}
\right.
\end{eqnarray}

\section{The Boubaker polynomials}

The first monomial definition of the Boubaker polynomials \cite{Awojoyogbe,Boubaker1,Chaouachi,Labiadh1} appeared in a physical study that yielded an analytical solution to heat equation inside a physical model \cite{Labiadh2,Slama1}. This monomial definitions is traduced by (\ref{eq:2.1}):
\begin{definition}
A monomial definition of the Boubaker polynomials is:
\begin{eqnarray}\label{eq:2.1}
B_n(X) \defeq \sum^{\xi(n)}_{p=0} \Br{\frac{n-4p}{n-p}\binom{p}{n-p}}(-1)^p X^{n-2p},
\end{eqnarray}
where $\xi(n) = \lfloor\frac n 2\rfloor \defeq \frac{2n+(-1)^n -1}4$ (The symbol $\lfloor*\rfloor$ designates the floor function). Their coefficients could be defined through a recursive formula (\ref{eq:2.2}):
\begin{eqnarray}\label{eq:2.2}
\left\{\begin{array}{rcl}
         B_n(X) & = & \sum^{\xi(n)}_{j=0}\Br{b_{n,j}X^{n-2j}}, \\
         b_{n,0} & = & 1, \\
         b_{n,1} & = & -(n-4), \\
         b_{n,j+1} & = & \frac{(n-2j)(n-2j-1)}{(j+1)(n-j-1)}\cdot\frac{n-4j-4}{n-4j}\cdot b_{n,j}, \\
         b_{n,\xi(n)} & = & \left\{\begin{array}{rl}
                                     (-1)^{\frac n2}\cdot 2 & \text{if\ } n \text{\ even} \\
                                     (-1)^{\frac {n+1}2}\cdot (n-2) & \text{if\ } n \text{\ odd}
                                   \end{array}
         \right.
       \end{array}
\right.
\end{eqnarray}
\end{definition}

\begin{definition}
A recursive relation which yields the Boubaker polynomials is:
\begin{eqnarray}\label{eq:2.3}
\left\{\begin{array}{rcl}
         B_m(X) & = & XB_{m-1}(X) - B_{m-2}(X), \text{for\ } m>2,\\
         B_2(X) & = & X^2 + 2,\\
         B_1(X) & = & X,\\
         B_0(X) & = & 1.
       \end{array}
\right.
\end{eqnarray}
\end{definition}

\section{Riordan matrices of the Boubaker polynomials}

In this section, we will present Riordan matrices analysis of the Boubaker polynomials. The notations and the results of \cite{Luzon1,Luzon2,Luzon3,Luzon4} will be used extensively.
We start with the following relation (demonstrated on page 25 in \cite{Luzon4}):
\begin{eqnarray}\label{eq:3.1}
B_n(x) = U_n\Pa{\frac x2} + 3U_{n-2}\Pa{\frac x2},\quad \text{for\ } n\geqslant2
\end{eqnarray}
then:
\begin{eqnarray}
B_{2m}(x) &=& U_{2m}\Pa{\frac x2} + 3U_{2m-2}\Pa{\frac x2}\nonumber \\
&=& 2\sum^m_{k=0}\widetilde{T}_{2k}\Pa{\frac x2} + 6\sum^{m-1}_{k=0}\widetilde{T}_{2k}\Pa{\frac x2}\label{eq:3.2}\\
&=& 8\sum^{m-1}_{k=0}\widetilde{T}_{2k}\Pa{\frac x2} + 2\widetilde{T}_{2m}\Pa{\frac x2} = 4+ 8\sum^{m-1}_{k=0}T_{2k}\Pa{\frac x2} + 2T_{2m}\Pa{\frac x2}\label{eq:3.3}.
\end{eqnarray}
In a similar way:
\begin{eqnarray}
B_{2m+1}(x) &=& 8\sum^{m-1}_{k=0}\widetilde{T}_{2k+1}\Pa{\frac x2} + 2\widetilde{T}_{2m+1}\Pa{\frac x2} = 8\sum^{m-1}_{k=0}T_{2k+1}\Pa{\frac x2} + 2T_{2m+1}\Pa{\frac x2}\label{eq:3.4}\\
 &=& 8\sum^{m-1}_{k=0}\widetilde{T}_{2k}\Pa{\frac x2} + 2\widetilde{T}_{2m+1}\Pa{\frac x2}\label{eq:3.5}
\end{eqnarray}
so:
\begin{eqnarray}
B_{2m}(2\cos t) &=& 4+ 8\sum^{m-1}_{k=1}T_{2k}(\cos t) + 2T_{2m}(\cos t) \nonumber\\
&=& 4+ 8\sum^{m-1}_{k=1}\cos(2kt) + 2\cos(2mt)\label{eq:3.6}\\
B_{2m+1}(2\cos t) &=& 8\sum^{m-1}_{k=1}\cos((2k+1)t) + 2\cos((2m+1)t)\label{eq:3.7}.
\end{eqnarray}
Now, consider another new polynomial class defined by:
\begin{eqnarray}
B_n(2\cos t) &=& \frac{B_n(2\cos t) - 2T_n(\cos t)}4,\quad n>1\label{eq:3.8}
\end{eqnarray}
or:
\begin{eqnarray}\label{eq:3.9}
\left\{\begin{array}{ccc}
         B_n(x) & = & \frac{B_n(x)-2T_n\Pa{\frac x2}}4 \\
         x & = & 2\cos t
       \end{array}
\right.
\end{eqnarray}
So using Eq.~(\ref{eq:3.8}) and Eq.~(\ref{eq:3.9}) we get:
\begin{eqnarray}
B_{2m}(x) &=& \frac{B_{2m}(x) - 2T_{2m}\Pa{\frac x2}}4 = 1+ 2\sum^{m-1}_{k=0}T_{2k}\Pa{\frac x2},\label{eq:3.10}\\
B_{2m+1}(x) &=&\frac{B_{2m+1}(x) - 2T_{2m+1}\Pa{\frac x2}}4 = 2\sum^{m-1}_{k=0}T_{2k}\Pa{\frac x2}.\label{eq:3.11}
\end{eqnarray}
In order to obtain a generating function and to make a polynomial sequence (i. e. the degree is the subindex) we consider
$$
\widetilde{B}_n(x) = B_{n-2}(x).
$$
So, symbolically:
\begin{eqnarray}\label{eq:3.12}
\left[
  \begin{array}{c}
    \widetilde{B}_0(x) \\
    \widetilde{B}_1(x) \\
    \widetilde{B}_2(x) \\
    \widetilde{B}_3(x) \\
    \widetilde{B}_4(x) \\
    \widetilde{B}_5(x) \\
    M \\
  \end{array}
\right] = \left[
            \begin{array}{ccccccc}
              2& 0 & 0 & 0 & 0 & 0 & 0 \\
              0 & 2 & 0 & 0 & 0 & 0 & 0 \\
              2 & 0 & 2 & 0 & 0 & 0 & 0 \\
              0 & 2 & 0 & 2 & 0 & 0 & 0 \\
              2 & 0 & 2 & 0 & 2 & 0 & 0 \\
              0 & 2 & 0 & 2 & 0 & 2 & 0 \\
              M & M & M & M & M & M & O \\
            \end{array}
          \right]
\left[
  \begin{array}{c}
    \widetilde{T}_0(x) \\
    \widetilde{T}_1(x) \\
    \widetilde{T}_2(x) \\
    \widetilde{T}_3(x) \\
    \widetilde{T}_4(x) \\
    \widetilde{T}_5(x) \\
    M \\
  \end{array}
\right]
\end{eqnarray}
We can write this in terms of Riordan matrices in the next way:

\begin{eqnarray}\label{eq:3.13}
\sum_{n\geqslant0}\widetilde{B}_n(t) = T\Pa{\frac{2}{1-x^2}\big|1}T\Pa{\frac{1-x^2}4\big|\frac{1+x^2}2}T(2|2)\Pa{\frac1{1-tx}}.
\end{eqnarray}
or:
\begin{eqnarray}\label{eq:3.14}
\sum_{n\geqslant0}\widetilde{B}_n(t)x^n = T(1|1+x^2)\Pa{\frac1{1-tx}}.
\end{eqnarray}
In fact we have the Riordan matrix:
\begin{eqnarray}\label{eq:3.15}
T(1|1+x^2)
\end{eqnarray}
which is:
\begin{eqnarray}\label{eq:3.16}
\left[
  \begin{array}{ccccccc}
    1 & 0 & 0 & 0 & 0 & 0 & 0 \\
    0 & 1 & 0 & 0 & 0 & 0 & 0 \\
    -1 & 0 & 1 & 0 & 0 & 0 & 0 \\
    0 & -2 & 0 & 1 & 0 & 0 & 0 \\
    1 & 0 & -3 & 0 & 1 & 0& 0 \\
    0 & 3 & 0 & -4 & 0 & 1 & 0 \\
    M & M & M & M & M & M & O \\
  \end{array}
\right]
\end{eqnarray}
Hence, the few first $\widetilde{B}_n(x)$ are:
\begin{eqnarray}\label{eq:3.17}
\left\{\begin{array}{rcl}
         \widetilde{B}_0(x) & = & 1 \\
         \widetilde{B}_1(x) & = & x \\
         \widetilde{B}_2(x) & = & x^2-1 \\
         \widetilde{B}_3(x) & = & x^3-2x \\
         \widetilde{B}_4(x) & = & x^4-3x^2+1 \\
         \widetilde{B}_5(x) & = & x^5 -4x^3 + 3x
       \end{array}
\right.
\end{eqnarray}
with the recurrence (\ref{eq:3.18}).
\begin{eqnarray}\label{eq:3.18}
\widetilde{B}_n(x) = x\widetilde{B}_{n-1}(x)- \widetilde{B}_{n-2}(x),\quad n\geqslant 2.
\end{eqnarray}
Note that this recurrence is the same as that for the Boubaker polynomials but with different initial conditions. In fact the relation between both families of polynomials is given by
\begin{eqnarray}\label{eq:3.19}
T(1+3x^2|1+x^2) = T(1+3x^2|1)T(1|1+x^2).
\end{eqnarray}
Then, finally:
\begin{eqnarray}\label{eq:3.20}
B_n(x) = x\widetilde{B}_{n-1}(x) + 3\widetilde{B}_{n-2}(x),\quad n\geqslant2.
\end{eqnarray}

\section{Fermat-linked expressions}

Using inversion of Riordan matrices we can get $\widetilde{B}_n(x)$ each as combinations of Boubaker polynomials.

\begin{remark}
Comparing the recurrence (\ref{eq:3.20}) with the one of the
Chebyshev polynomials of the second kind, we can obtain an explicit
expression of the new polynomials defined by
(\ref{eq:3.8}-\ref{eq:3.9})
\begin{eqnarray}\label{eq:4.1}
B_n(x) = \frac{\sin((n+1)t)}{\sin t}, \quad x = 2\cos t, \quad n =0,
1,2,\ldots.
\end{eqnarray}
In another word, the new polynomial  is the scaled Chebyshev
polynomial $U_n(x)$ of the second kind, since the relation between
the two polynomials is related as:
\begin{eqnarray}\label{eq:4.2}
B_n(2x) = U_n(x), \quad n=0,1,2,\ldots.
\end{eqnarray}
\end{remark}

\begin{remark}
By using (\ref{eq:4.1}) or (\ref{eq:4.2}), we can obtain some other
relations. In fact Fermat polynomials are obtained by setting
$p(x)=3x$ and $q(x)=-2$ in the Lucas polynomial sequence, defined by
(\ref{eq:4.3}).
\begin{eqnarray}\label{eq:4.3}
F_n(x) = p(x)F_{n-1}(x)+q(x)F_{n-2}(x).
\end{eqnarray}
As A. Luzon and M. A. Moron \cite{Luzon1,Luzon2,Luzon3,Luzon4}
demonstrated, through the associated Riordan matrix:
\begin{eqnarray}\label{eq:4.4}
\left[
  \begin{array}{ccccccccc}
    \frac13 &  &  &  &  &  &  &  &  \\
    0 & 1 &  &  &  &  &  &  &  \\
    0 & 0 & 3 &  &  &  &  &  &  \\
    0 & -2 & 0 & 9 &  &  &  &  &  \\
    0 & 0 & -12 & 0 & 27 &  &  &  &  \\
    0 & 4 & 0 & -54 & 0 & 81 &  &  &  \\
    0 & 0 & 36 & 0 & -216 & 0 & 243 &  &  \\
    0 & -8 & 0 & 216 & 0 & -810 & 0 & 729 &  \\
    \vdots & \vdots & \vdots & \vdots & \vdots & \vdots & \vdots & \vdots & \ddots \\
  \end{array}
\right]
\end{eqnarray}
that
\begin{eqnarray}\label{eq:4.5}
\left\{\begin{array}{rcl}
         F_1(x) & = & 1 \\
         F_2(x) & = & 3x \\
         F_3(x) & = & 9x^2-2 \\
         F_4(x) & = & 27x^3-12x \\
         \cdots & \cdots & \cdots
       \end{array}
 \right.
\end{eqnarray}
and
\begin{eqnarray}\label{eq:4.6}
F_x(x) = \Pa{\sqrt{2}}^n U_n\Pa{\frac{3x}{2\sqrt{2}}}
\end{eqnarray}
\end{remark}

\begin{thrm}
Let $(R,+,o)$ be a commutative ring, $(D,+,o)$ be an integral domain
such that $D$ is a subring of $R$ whose zero is $0_D$ and whose
unity is $1_D$, $X\in R$ be transcendental over $D$, $D[X]$ be the
ring of polynomials forms in $X$ over $D$, and finally denote
Boubaker polynomials and Fermat polynomials as $B_n(x)$ and $F_n(x)$
,respectively, as polynomials contained in $D[X]$, then:
\begin{eqnarray}\label{eq:4.7}
B_n(x) =\frac1{(\sqrt{2})^n}F_n\Pa{\frac{2\sqrt{2}x}3} +
\frac1{(\sqrt{2})^{n-2}}F_{n-2}\Pa{\frac{2\sqrt{2}x}3};\quad n =
0,1,2,\ldots
\end{eqnarray}
\end{thrm}

\begin{proof}
Riordan matrices for Boubaker polynomials and Fermat polynomials
(see \cite{Luzon1,Luzon2,Luzon3,Luzon4}) are respectively:
\begin{eqnarray}\label{eq:4.8}
\sum^{+\infty}_{n=0}B_n(x)t^n =
(1+3x^2|1+x^2)\Pa{\frac1{1-xt}},\quad \sum^{+\infty}_{n=0}F_n(x)t^n
=\Pa{\frac13|\frac{1+x^2}3}.
\end{eqnarray}
Let's expand the inverse Riordan arrays:
\begin{eqnarray}\label{eq:4.9}
T(1+3x^2|1+x^2) = T(1+3x^2|1)T\Pa{\frac12|\frac{1+x^2}2}T(2|2),
\end{eqnarray}
which gives
\begin{eqnarray}\label{eq:4.10}
T(1+3x^2|1+x^2) =
T(1+3x^2|1)T(1|\sqrt{2})T\Pa{\frac13|\frac{1+x^2}3}T(3|\frac3{\sqrt{2}}).
\end{eqnarray}
By identifying Riordan matrix for Fermat polynomials in the right
term of Eq. (\ref{eq:4.10}), the desired equality holds.
\end{proof}

Expressions (\ref{eq:4.2}) and (\ref{eq:4.7}) are very useful for
developing the already proposed Boubaker polynomials Expansion
Scheme (BPES).

\section{Conclusion}

The Boubaker polynomials have been  investigated. Using y Riordan
matrices analysis, a sequence of relations outlining the relations
with Chebyshev and Fermat polynomials have been obtained as guides
to further studies. The obtained expression are a meaningful supply
to recent applied physics studies
\cite{Mahmoud1,Lazzez,Fridjine1,Khelia,Mahmoud2,Dada,Rahmanov1,Rahmanov2,Tabatabaei,Fridjine2,Belhadj1,Belhadj2,Zhang}
using the Boubaker polynomials Expansion Scheme (BPES).



\begin{thebibliography}{99}
\bibitem{Rahmanov1} H. Rahmanov, Triangle read by rows: row $n$ gives coefficients of Boubaker polynomial $B_n(x)$, calculated for X=2cos(t), centered by adding -2cos(nt), then divided by 4, in order of decreasing exponents. OEIS (Encyclopedia of Integer Sequences), A160242
\end{thebibliography}
\end{document}